\documentclass[a4paper]{llncs}

\usepackage{amsthm,amsfonts}
\usepackage{amsmath}
\usepackage{amssymb}
\usepackage{graphicx}
\usepackage{makeidx}
\usepackage{enumitem}
\usepackage{color}
\usepackage{mathtools}

\usepackage{pifont}

\usepackage{bm}

\newtheorem{mydef}{Definition}
\newtheorem{myteo}{Theorem}
\newtheorem{mypro}{Proposition}
\newtheorem{mycor}{Corollary}
\newtheorem{mylem}{Lemma}
\newtheorem{myexe}{Example}

\newtheorem{myalg}{Algorithm}

\DeclareMathOperator*{\argmin}{arg\,min}
\DeclareMathOperator*{\argmax}{arg\,max}

\begin{document}

\title{Channel Metrization} 

\author{Rafael Gregorio Lucas D'Oliveira \and Marcelo Firer} 
\institute{Imecc - Unicamp\\ 
\email{rgldoliveira@gmail.com, mfirer@ime.unicamp.br }}

\maketitle

\begin{abstract}

We present an algorithm that, given a channel, determines if there is a distance for it such that the maximum likelihood decoder coincides with the minimum distance decoder. 

We also show that any metric, up to a decoding equivalence, can be isometrically embedded into the hypercube with the Hamming metric, and thus, in terms of decoding, the Hamming metric is universal.
\end{abstract}

keywords: Coding Theory; Hypercube Embeddings; Maximum-likelihood Decoding; Minimum Distance Decoding. 

\section{Introduction}

The binary symmetric channel (BSC) is probably the most commonly used discrete channel in Coding and Information Theory, not only because of its simplicity, but also because it represents the worst case scenario, when there is no information about the channel. This gives the Hamming metric a prominent status since it is matched to the BSC (the minimum distance decoder is equivalent to the maximum likelihood decoder). Less attention was given to other metrics (see \cite{gabi12} and \cite{deza14}[Chapter 16] for applications of other metrics in coding theory).

The general problem of determining if a channel is matched to a metric \footnote{This problem was first asked by Massey in \cite{mass67}.}, or in other words, if a channel is metrizable (in analogy to metrizable spaces in topology) has been very little explored (see \cite{segu80} \footnote{S\'{e}guin considers families of metrics that are defined over an alphabet and which extend additively over the coordinates, which include the Hamming and Lee metrics.} and \cite{fire14} \footnote{Walker and Firer proved that a metric can be matched to the Z-Channel, which in a sense is the most anti-symmetrical of all channels.}). When this occurs, it gives more structure to the channel, and it might be useful for, among other things, the construction of efficient codes since with distances, come other related concepts such as: packing radius, perfect codes, MDS codes, etc.

We present mainly two results: An algorithm which determines if a channel is metrizable; and show that the Hamming metric is universal, in the sense that, for any metrizable channel, this metric can be used for decoding purposes.
\newline

This paper is organized in the following manner:

Section 2 presents the key notion of this paper: decoding equivalence.

In Section \ref{match} we present a metrization algorithm for channels. 

In Section \ref{sec:setpatt}, we introduce set patterns, a generalization of intersection patterns. This is the main tool to prove the theorems of Section \ref{emb}. 

In Section \ref{emb} we present our second main result. We show that any metric may be isometrically embedded (up to a decoding equivalence) into the hypercube with the Hamming metric, which we will refer to as the Hamming cube. Thus, in terms of decoding, the Hamming metric is universal (this is analogous to embedding theorems in differential topology). This means that for any metrizable space we can use the Hamming metric for minimum distance decoding. 

In Section \ref{geo} we introduce a geometrical approach which can be used to find minimal dimensional embeddings and which we believe warrants further investigation. This section contains preliminary results on future work.

\section{Distances, Channels and Decoders}

We distinguish metrics from distances (we follow the notation from \cite{deza14}).

Let $X$ be a set. A \emph{distance} is a function $d: X \times X \mapsto \mathbb{R}_{\geq 0}$ such that
\begin{enumerate}
\item $d(x,x) = 0$
\item $d(x,y) = d(y,x)$
\end{enumerate}

If the distance satisfies property 3 it is a \emph{semimetric} \footnote{There is no general consensus on this term. In topology, for example this is sometimes called a symmetric.}. If it also satisfies the triangle inequality (property 4), then it is called a \emph{metric}.

\begin{enumerate}
\setcounter{enumi}{2}
\item $d(x,y)=0 \Rightarrow x=y$
\item $d(x,z) \leq d(x,y) + d(y,z)$
\end{enumerate}

If the set $X$ is an abelian group and $d(x+z,y+z) = d(x,y)$
we say that $d$ is \emph{translation invariant}. In this case the distance is determined by a \emph{weight function}  $\omega : X \mapsto \mathbb{R}_{\geq 0}$, defined by $\omega(x): = d(x,0)$, so that $d(x,y) = \omega(x-y)$.

We are concerned only with finite sets, thus nothing is lost if we assume $X$ to be the set $[n] = \{ 1,2,3,\ldots,n \}$. We also identify the distance with its matrix representation $d_{n \times n}$ such that $d_{ij} = d(i,j)$.

We define a channel using the notation of Probability Theory.

\begin{mydef}\label{def:channel}
A channel over $[n]$ is an $n\times n$ probability matrix $P$ such that
\[P_{ij} = P(\text{$j$ received } | \text{$i$ sent}) .\]
\end{mydef}

Thus, a channel is a matrix $P_{n \times n}$ with entries in $[0,1]$ such that
$\sum_j P_{ij} = 1$.

As usual, in Information Theory, the interpretation is that a transmitter sends a symbol $i \in [n]$ to a receiver and the channel determines the probability that $j \in [n]$ is received. 

We only consider ``square'' channels, i.e. those in which the size of the sets of inputs is equal to the size of the sets of outputs, and, consequently, the probability matrix is square \footnote{The probabilities are defined over the set of elements of $X$, that may be considered as a set of all possible messages and not an alphabet over which messages are written. }. 

Given a \emph{code} $C \subseteq [n]$, a \emph{maximum likelihood decoder} (MLD) is such that $j \in [n]$ is decoded as an element of the set $\argmax \{ P(j|c):c\in C \}$.

If a distance is defined on $[n]$, a \emph{minimum distance decoder} (MDD) is such that $j \in [n]$ is decoded as an element of the set $\argmax \{ d(c,j): c\in C \}$.

Given a distance $d$ and a channel over $[n]$ we say that they are\emph{ matched} if, for every $C\subseteq [n]$ and $j\in [n]$,
\[
\argmin \{d(j,c);c\in C \} = \argmax \{ P(j|c);c\in C \},
\] 
i.e. if the MLD and the MDD coincide.

It is well known that the Hamming metric is matched to the binary symmetric channel. Very little is known about the case for general channels.

\section{Finding a Distance Matched to a Channel} \label{match}

We now present the key notion of this paper.

\begin{mydef}
Two distances $d_1$ and $d_2$ over $[n]$ are decoding equivalent, denoted by $d_1 \sim d_2$, if they define the same minimum distance decoder.

In technical terms, $d_1 \sim d_2$ if for any code $C\subseteq [n]$ and any $j\in [n]$,
\[ \argmin\{d_1(c,j): c\in C \}= \argmin \{d_2(c,j): c\in C \} .\]
\end{mydef}

In fact, $d \sim d'$ if and only if they preserve the \emph{weak ordering} (ordering allowing ties) of the distances from a fixed point.

\begin{mypro} \label{prop:ineq}
Let $d_1$ and $d_2$ be two distances over $[n]$. Then, $d_1 \sim d_2$ if and only if for every $i,k,j \in [n]$, it holds that
\[ d_1 (i,j) < d_1 (k,j) \Leftrightarrow d_2 (i,j) < d_2 (k,j) .\]
\end{mypro}

\begin{proof}
Let $d_1 \sim d_2$. Take $j \in [n]$ and suppose that there are $i,k \in [n]$ such that 
$ d_1 (i,j) < d_1 (k,j) $.
Consider the code $C =\{ i,k \}$. Then,
\[ \argmin\{d_1(c,j);c\in C \} = \{i \} .\]
Since $d_1 \sim d_2$, it follows that
\[ \argmin\{d_2(c,j);c\in C \} = \{i \} ,\]
and therefore, $d_2 (i,j) < d_2 (k,j)$.

In the case that that $d_2 (i,j) < d_2 (k,j)$, the analogous arguments would show that $d_1 (i,j) < d_1 (k,j)$.
Hence,
\[ d_1 (i,j) < d_1 (k,j) \Leftrightarrow d_2 (i,j) < d_2 (k,j) .\]

Now suppose that for every $i,k,j \in [n]$
\[ d_1 (i,j) < d_1 (k,j) \Leftrightarrow d_2 (i,j) < d_2 (k,j) .\] Let $C$ be a code and fix $j \in [n]$. Suppose that 
\[  \argmin\{d_1(c,j): c\in C \} \neq \argmin \{d_2(c,j): c\in C \} .\]
Without loss of generality, we can suppose that there exists a $k \in C$ such that
\[ k \in \argmin\{d_1(c,j): c\in C \}, \hspace{5pt} k \notin \argmin \{d_2(c,j): c\in C \}, \hspace{5pt} \text{and} |C| \geq 2 .\]

Let $i \in \argmin \{d_2(c,j): c\in C \}.$ Then, it holds that $d_2(i,j) < d_2 (k,j)$. 
But this implies, by equivalence, that $d_1(i,j) < d_1 (k,j)$,
and therefore, since $i \in C$,
\[ k \notin \argmin\{d_1(c,j): c\in C \} ,\]
a contradiction. Finalizing the proof.
\end{proof}

Proposition \ref{prop:ineq} implies directly that two equivalent distances will have the same set of balls, not necessarily for the same radiuses.

\begin{mycor}
Let $d_1$ and $d_2$ be two distances over $[n]$. Then, $d_1 \sim d_2$ if and only if for every $x_0 \in X$ and $r_1 \in \mathbb{R}$ there exists an $r_2=r_2(x_0,r_1) \in \mathbb{R}$ such that
$B_{d_1}(x_0,r_1) = B_{d_2}(x_0,r_2)$, where $B_d(x,r) = \{ y \in X : d(y,x) \leq r \}$.
\end{mycor}

In terms of the matrix representation two distances are equivalent if the weak orderings of the elements of each column are the same. Since this is also true for channels under maximum likelihood decoding, we will define the following concept for general matrices.

\begin{mydef}
Given a matrix $M$, its decreasing column weak ordered matrix is the matrix $O^- M$ such that $(O^- M)_{ij}=k$ if $M_{ij}$ is the $k$th largest element (allowing ties) of the $j$th column.

Similarly, $(O^+ M)_{ij}=k$ if it is the $k$th smallest element of the $j$th column.
\end{mydef}

\begin{myexe}
If
$ M = 
\begin{pmatrix}
9 & 2 & 1 \\ 
9 & 7 & 0 \\ 
8 & 6 & 8
\end{pmatrix}
$,
then
\[O^- M = 
\begin{pmatrix}
1 & 3 & 2 \\ 
1 & 1 & 3 \\ 
2 & 2 & 1
\end{pmatrix}
\hspace{5mm}
\text{and}
\hspace{5mm}
O^+ M =
\begin{pmatrix}
2 & 1 & 2 \\ 
2 & 3 & 1 \\ 
1 & 2 & 3
\end{pmatrix}
\]

\end{myexe}

Note that $O^- M = O^- N$ if and only if $O^+ M = O^+ N$.

\begin{mycor} \label{pro:equiv}
Let $d_1$ and $d_2$ be two distances over $[n]$. Then, $d_1 \sim d_2$ if and only if $O^- d_1 = O^- d_2$ or equivalently, if $O^+ d_1 = O^+ d_2$.
\end{mycor}

The distances used in Coding Theory are usually metrics. However, every semimetric is decoding equivalent to some metric as shown by the following distance transform (a particular case of the "squeezing" argument in \cite{fire14}).

\begin{myexe} \label{exa:metric}
Let $d_1$ be a semimetric. Consider the distance transform
\[ d_2 (x,y)  = 1 + \frac{d_1 (x,y)}{\max_{u,v} d_1(u,v)} \]
for $x \neq y$ and zero otherwise. Then, $d_1 \sim d_2$ and $d_2$ is a metric.
\end{myexe}

Analogously to the case for distances, we have the following equivalence relation between channels.

\begin{mydef}
Two channels, $M$ and $N$, are called decoding equivalent, $M \sim N$, if for any code $C \subset X$, they define the same maximum likelihood decoder.
\end{mydef}

The maximum likelihood decoder and the minimum distance decoder are essentially the same. The only difference is that with the MLD we are searching for the largest entry in a column (restricted to the rows corresponding to codewords) while with the MDD we are looking for the smallest entry. In both cases, only the weak ordering of the elements in the columns are important. 

We first show that this is also the case for channels.

\begin{mypro}
Let $M$ and $N$ be two channels over $[n]$. Then, $M \sim N$ if and only if for every $i,k,j \in [n]$
\[ M_{i,j} < M_{k,j} \Leftrightarrow N_{i,j} < N_{k,j} .\]
\end{mypro}

\begin{proof}
The proof is analogous to that of Proposition \ref{prop:ineq}.
\end{proof}

\begin{mycor}
Let $M$ and $N$ be two channels over $[n]$. Then, $M \sim N$ if and only if $O^- M = O^- N$ or equivalently $O^+ M = O^+ N$.
\end{mycor}

Thus, determining if a channel $M$ is matched to a distance $d$ is equivalent to determining if there exists a distance $d$ such that $O^- M = O^+ d$.

This is equivalent to solving the following system of inequalities on $d(i,j)$:
\[ \left\{\begin{matrix}
d(i,j) = d(k,j) & \text{if } O^- M_{ij}= O^- M_{kj} \\ 
d(i,j) < d(k,j) & \text{if } O^- M_{ij}< O^- M_{kj}
\end{matrix}\right.
\]
and the symmetry conditions $d(i,j) = d(j,i)$.

Since it is a distance we also need that $d(i,i) = 0$, i.e. $O^- M_{ii}= 1$. To find a metric we we must also require that the distance be a semimetric so that we can apply a transformation such as in Example \ref{exa:metric}. For this to happen only the diagonals of $O^- M_{ii}$ can have ones.

We present the following algorithm which given a $n \times n$ matrix $M$ determines if there exists (and in this case constructs) a distance $d$ on $[n]$ such that $O^- M = O^+ d$ or if no such distance exists (and in this case outputs an inequality showing why this is not possible).

\begin{myalg}
We have $n$ chains of inequalities, each corresponding to a certain column. The smallest element of each chain must be $d(i,i)$ which we set to equal zero. Then:

\begin{enumerate}
\item We take the first chain and set arbitrary values for the distances with the condition that the inequalities hold true; 
\item We then set the same values to their corresponding symmetric term, i.e. $d(i,j)=d(j,i)$; 
\item We continue to do this until we have assigned a value to every distance and have therefore found a distance matched to our channel; \\
or,
\end{enumerate}
\begin{enumerate}[label=\arabic*'.,ref=(\arabic*'),start=3]
\item Find that some distance cannot have a valued assigned to it, and thus there is no distance matched for this channel. 
\end{enumerate}
\end{myalg}

\begin{proof}
The only way for the algorithm not to work is if the choices made in step 1 could influence our result. This is not possible because of Proposition \ref{prop:ineq}, where it is shown that only the weak ordering of the values matter.
\end{proof}

We illustrate the use of the algorithms in the following two examples.

\begin{myexe}
Let 
$ M =
\begin{pmatrix}
\frac{5}{8} & \frac{1}{8} & \frac{2}{8} \\[1ex]
\frac{2}{8} & \frac{5}{8} & \frac{1}{8} \\[1ex]
\frac{1}{8} & \frac{2}{8} & \frac{5}{8}
\end{pmatrix}
$.
Then,
$O^- M =
\begin{pmatrix}
1 & 3 & 2 \\ 
2 & 1 & 3 \\ 
3 & 2 & 1
\end{pmatrix}
.$

We have the following three chains of inequalities (corresponding to the columns):
\[ 0=d(1,1) < d(1,2) < d(1,3)\]
\[ 0=d(2,2) < d(2,3) < d(2,1)\]
\[ 0=d(3,3) < d(3,1) < d(3,2)\]
We set arbitrary values to the elements in the first column and the same to their symmetric counterparts (since a distance is symmetric).
\[ 0=d(1,1) < 1=d(1,2) < 2=d(1,3)\]
\[ 0=d(2,2) < d(2,3) < 1=d(2,1)\]
\[ 0=d(3,3) < 2=d(3,1) < d(3,2)\]
In the next step we must set an arbitrary value to $d(2,3)$ but it is impossible to do this since it must satisfy
\[ 2=d(3,1) <d(2,3) < 1=d(2,1).\]
Therefore, $M$ is not metriziable since the following is not possible
\[ d(2,1) < d(3,1) <d(2,3) < d(2,1) .\]
\end{myexe}

\begin{myexe}
Let 
$ M =
\begin{pmatrix}
\frac{5}{8} & \frac{3}{16} & \frac{3}{16} \\[1ex] 
\frac{1}{4} & \frac{1}{2} & \frac{1}{4} \\[1ex] 
\frac{1}{8} & \frac{2}{8} & \frac{5}{8}
\end{pmatrix}
$.
Then,
$O^- M =
\begin{pmatrix}
1 & 3 & 3 \\ 
2 & 1 & 2 \\ 
3 & 2 & 1
\end{pmatrix}
$.

We have the following three chains of inequalities:
\[ 0=d(1,1) < d(1,2) < d(1,3)\]
\[ 0=d(2,2) < d(2,3) < d(2,1)\]
\[ 0=d(3,3) < d(3,2) < d(3,1)\]
We set arbitrary values to the first chain and to their symmetric counterparts.
\[ 0=d(1,1) < 1=d(1,2) < 2=d(1,3)\]
\[ 0=d(2,2) < d(2,3) < 1=d(2,1)\]
\[ 0=d(3,3) < d(3,2) < d(3,1)=2\]
We do the same for the second chain.
\[ 0=d(1,1) < 1=d(1,2) < 2=d(1,3)\]
\[ 0=d(2,2) < d(2,3)=\frac{1}{2} < 1=d(2,1)\]
\[ 0=d(3,3) < d(3,2)=\frac{1}{2} < d(3,1)=2\]
We were able to set values to all the distances. Therefore, $M$ is matched to the following distance:
\[d =
\begin{pmatrix}
0 & 1 & 2 \\ 
1 & 0 & \frac{1}{2} \\ 
2 & \frac{1}{2} & 0
\end{pmatrix}
.\]

Note that this distance is not a metric since
\[ 2 = d(1,3) > d(1,2) + d(2,3) = 1 + \frac{1}{2} .\]
Since $O^- M$ has ones only in the diagonal we can apply the transformation of Example \ref{exa:metric} to get a metric matched to the channel. In this case the metric
\[d_2 =
\begin{pmatrix}
0 & \frac{3}{2} & 2 \\ 
\frac{3}{2} & 0 & \frac{5}{4} \\ 
2 & \frac{5}{4} & 0
\end{pmatrix}
\]
is matched to the channel $M$.

\end{myexe}

In terms of algorithmic complexity, if our matrix is $n\times n$ we have at most $n^2$ elements to set and for each of these we make at most $2n$ comparisons (each chain of inequalities has size $n$), giving a trivial upper bound of $O(n^3)$ on the number of operations to be performed.

In application the size of the matrix is usually exponential.

\section{Set Patterns} \label{sec:setpatt}

In this section we present Theorem \ref{teo:sys}, the main tool for proving our main embedding theorems, Theorems \ref{teo:translation embedding} and \ref{teo:embedding}. Their proof consists in reducing the problem into determining if there exists sets satisfying certain properties.

In \cite{deza73}, it is shown that isometrically embedding a distance into the Hamming cube is equivalent to solving a problem of the following type\footnote{Many combinatorial problems can be stated in this form (see the references in \cite{chva80}).}: given an $n \times n$ matrix $A=(a_{ij})$, decide wether there are exists sets $S_1, S_2, \ldots , S_n$ such that $| S_i \cap S_j | = a_{ij}$ for every $i,j \in [n]$.

The matrix $A$ is known as an \emph{intersection pattern}, and when such sets exist, the intersection pattern is said to be \emph{realizable}.

Determining if an intersection problem is realizable is $NP$-complete \cite{chva80} and, therefore, so is determining if a distance is isometrically embeddable into the Hamming cube.

To prove our results we will need to generalize on the notion of intersection patterns by defining what we call \emph{set patterns}.

We tacitly assume that all subsets $I$ considered in the sequel are nonempty.
 
\begin{mydef}
Given a finite family of sets $\mathcal{F} = \{ A_1, A_2, \ldots , A_n \}$ a \emph{minterm}\footnote{This term is taken from Boolean algebra.} is a set  $X_I = \{ a \in \cup_{i=1}^n A_i : i \in I \Rightarrow a \in A_i, i \notin I \Rightarrow a \notin A_i, \forall I \subseteq [n] \}$, for every subset $I \subseteq [n]$. The cardinalities of the minterms are denoted by lowercase letters $x_I = | X_I |$.
\end{mydef}

The minterms are the disjoint components of the Venn diagram of the sets.

\begin{figure}
  \centering
      \includegraphics[width=5cm]{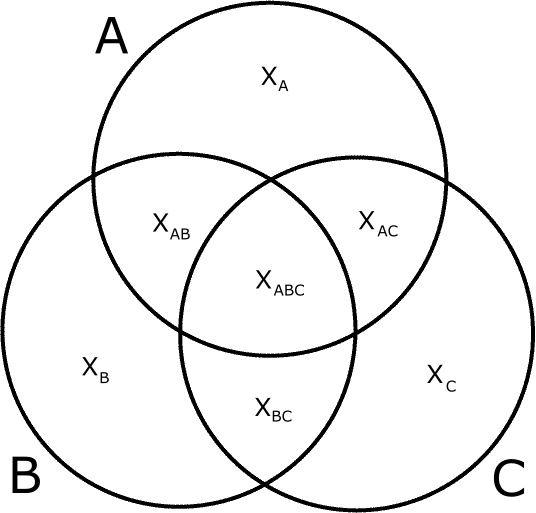}
  \caption{\label{fig:venn} Minterms of the family $\mathcal{F} = \{A,B,C \}$.}
\end{figure}

We now present the main notion of this section.

\begin{mydef}
A \emph{set pattern} is a triple $(G,c,n)$ where $n$ is a positive integer, $G : \mathbb{R}^{2^n - 1} \rightarrow \mathbb{R}^m$, and $c \in \mathbb{R}^m$. An $x \in \mathbb{R}^{2^n - 1}$ is called a \emph{solution} of the\footnote{We will use subsets as indexes and leave unspecified, but assume as given, the bijection from $[2^{n} - 1]$ to $2^{[n]} - \emptyset$. } set pattern if $G(x) = c$.

If there exists a finite family of sets $\mathcal{F} = \{ A_1, A_2, \ldots , A_n \}$ such that the cardinalities of the minterms are a solution of the set pattern, we say that $\mathcal{F}$ is a \emph{realization} of the set pattern and that the set pattern is \emph{realizable}.
\end{mydef}

A set pattern is, essentially, a system of equations for which integer solutions correspond to a family of sets satisfying the pattern imposed by the equations.

\begin{mypro} \label{pro:realizable}
A set pattern $(G,c,n)$ is realizable if and only if there exists $x \in  \mathbb{Z}^{2^n - 1}_{\geq 0}$ such that $G ( x ) = c$.
\end{mypro}

\begin{proof}
It follows directly from the definition.
\end{proof}

\begin{myexe}
Does there exist sets $A_1, A_2, A_3$ such that $|A_1 \bigtriangleup A_2 | = 3$, $|A_3|^{|A_1|} = 27$ and $|A_1 \cap A_2|^2 = 9$?

This problem is equivalent to the realizability of the following set pattern:
\[
\left\{\begin{matrix*}[l]
G_1 (x) = x_1 + x_2 + x_{13} + x_{23} & = 3 & = c_1 \\ 
G_2 (x) =(x_3 + x_{13} + x_{23} + x_{123})^{x_1+x_{12} + x_{13} + x_{123}} & = 27 & = c_2 \\ 
G_3 (x) =(x_{12} + x_{123})^2 & = 9 & = c_3
\end{matrix*}\right.
\]

Since $(x_1, x_2, x_3, x_{12}, x_{13}, x_{23}, x_{123}) = (2,0,1,2,0,1,1) \in \mathbb{Z}^{7}_{\geq 0}$ is a solution, it follows from Proposition \ref{pro:realizable} that such sets exist.
\end{myexe}

We are only interested in set patterns which correspond to intersections and symmetric differences. We must generalize these functions to $\mathbb{R}^{2^n - 1}$.

We start by generalizing set intersections.

\begin{mydef}\label{def:inter}
Let $J \subseteq [n]$. The \emph{$J$-wise intersection function} is defined as  $\Cap_J : \mathbb{R}^{2^n - 1} \rightarrow \mathbb{R}$ such that

\[ \Cap_J (x) = \sum_{J \subseteq I \subseteq [n]} x_I    .\]

The $J$-wise intersection functions, where $|J| \leq k$, is denoted by $\bm{\Cap_k} = ( \Cap_{J \subseteq [n]} )_{ |J| \leq k}$.
\end{mydef}

It is easy to see that $J$-wise intersections are linear and that, moreover, $\bm{\Cap_n}$ is a linear automorphism\footnote{A linear transformation from $\mathbb{R}^{2^n - 1}$ to itself.}, and thus, has a unique solution. 

\begin{myexe}
Does there exist sets $A_1, A_2, A_3$ such that

\[
\begin{matrix}
|A_1| = 6 \hspace{5mm}& |A_1 \cap A_2| = 6 \hspace{5mm}& |A_1 \cap A_2 \cap A_3| = 4 \hspace{5mm}\\ 
|A_2| = 9 \hspace{5mm}& |A_1 \cap A_3| = 5 \hspace{5mm}& \hspace{5mm}\\ 
|A_3| = 8 \hspace{5mm}& |A_2 \cap A_3| = 7 \hspace{5mm}& \hspace{5mm}
\end{matrix}
?\]

This problem is equivalent to the realizability of the following set pattern:
\[
\begin{matrix}
\Cap_{1}(x) = 6 \hspace{5mm}& \Cap_{12}(x) = 6 \hspace{5mm}& \Cap_{123}(x) = 4 \hspace{5mm}\\ 
\Cap_{2}(x) = 9 \hspace{5mm}& \Cap_{13}(x) = 5 \hspace{5mm}& \hspace{5mm}\\ 
\Cap_{3}(x) = 8 \hspace{5mm}& \Cap_{23}(x) = 7 \hspace{5mm}& \hspace{5mm}
\end{matrix}
\]
or equivalently, $\bm{\Cap_n} (x) = (6,9,8,6,5,7,4)$. The unique solution is $ x = (-1,0,0,2,1,3,4) \notin \mathbb{Z}^{7}_{\geq 0}$ and thus, by Proposition \ref{pro:realizable}, no such sets exist.

\end{myexe}

Intersection patterns are a particular case of set patterns of the form $(\bm{\Cap_2}, c, n)$.

Given a family $\mathcal{F} = \{ A_1, A_2, \ldots , A_n \}$, we have the following known relation

\begin{equation} \label{eq:dif}
 |\triangle_{i \in I} A_i| = \sum_{l=1}^{|I|} (-2)^{l-1} \sum_{ \substack{K \subseteq I \\ |K| = l}} |\cap_J A_{i \in J}| .
\end{equation}

We use equation \ref{eq:dif} to generalize symmetric differences.

\begin{mydef} \label{def:symm}
Let $J \subseteq [n]$. The \emph{$J$-wise symmetric difference function} is defined as the function $\blacktriangle_J : \mathbb{R}^{2^n - 1} \rightarrow \mathbb{R}$ such that

\[ \blacktriangle_J (x) = \sum_{l=1}^{|J|} (-2)^{l-1} \sum_{\substack{K \subseteq I \\ |K| = l}} \Cap_K (x) \]

The $J$-wise symmetric difference functions, where $|J| \leq k$, is denoted by $\bm{\blacktriangle_k} = (\blacktriangle_{J \subseteq [n]} )_{ |J| \leq k}$.
\end{mydef}

It follows directly from the linearity of $J$-wise intersections that $J$-wise symmetric differences are linear and that $\bm{\blacktriangle_n}$ is a linear automorphism.

\begin{myexe} \label{exe:sym}
Does there exist sets $A_1, A_2, A_3$ such that

\[
\begin{matrix}
|A_1| = 3 \hspace{5mm}& |A_1 \bigtriangleup A_2 | = 3 \hspace{5mm}& |A_1 \bigtriangleup A_2 \bigtriangleup A_3 | = 3 \\ 
|A_2| = 2 \hspace{5mm}& |A_1 \bigtriangleup A_3 | = 3 \hspace{5mm}& \\ 
|A_3| = 1 \hspace{5mm}& |A_2 \bigtriangleup A_3 | = 2 \hspace{5mm}& 
\end{matrix}
?\]

This problem is equivalent to the realizability of $\bm{\blacktriangle_n} (x) = (3,2,1,3,3,2,3)$.

By Definition \ref{def:symm}, we can recursively calculate the $J$-wise intersections and show that the problem is equivalent to the realizability of $\bm{\Cap_n} (x) = (3,2,1,1,\frac{1}{2},\frac{1}{2},\frac{1}{4})$, which has as a unique solution $x = \frac{1}{4} (7,3,1,3,1,1,1) \notin \mathbb{Z}^{7}_{\geq 0}$. 

Thus, no such sets exist.
\end{myexe}

We denote by $\bm{\overrightarrow{1}}$ the vector $(1,1,\ldots , 1)$ with all entries equal to one.

The following Lemma will be essential for proving Theorem \ref{teo:sys}.

\begin{mylem} \label{lem:novo} It holds that
\[\bm{\blacktriangle_n} ( \bm{\overrightarrow{1}} ) = 2^{n-1} \bm{\overrightarrow{1}} .\]
\end{mylem}

\begin{proof}
We start by calculating $ \Cap_J (\bm{\overrightarrow{1}}) = \sum_{I \in [n]} [J \subseteq I] = 2^{n-|J|}.$
Thus,
\begin{align*} 
\blacktriangle_J ( \bm{\overrightarrow{1}} ) &=  \sum_{l=1}^{|I|} (-2)^{l-1} \sum_{\substack{K \subseteq I \\ |K| = l}} \Cap_K (\bm{\overrightarrow{1}}) 
=  \sum_{l=1}^{|I|} (-2)^{l-1} \sum_{\substack{K \subseteq I \\ |K| = l}} 2^{n-|K|} \\ 
&=  \sum_{l=1}^{|I|} (-2)^{l-1} \binom{|I|}{l} 2^{n-l}  
=  2^{n-1} \sum_{l=1}^{|I|} \binom{|I|}{l} (-1)^{l-1}   \\ 
&= 2^{n-1},\\
\end{align*}
where in the last equality we use the identity $\sum_{i=1}^k \binom{k}{i} (-1)^{i-1} = 1$.

\end{proof}

We are now ready to prove the main result of this section.

\begin{myteo} \label{teo:sys}

Given $c \in \mathbb{Q}^{2^n-1}_+$, there exists  positive integers $m$ and $k$ such that the set pattern $(\bm{\blacktriangle_n},(m c + \bm{\overrightarrow{k}}),n)$ is realizable.

\end{myteo}

\begin{proof}
Let $x$ be the solution of $(\bm{\blacktriangle_n},c,n)$. By definitions \ref{def:inter} and \ref{def:symm}, the rationality of the $x_I$ follows from that of the $c_I$.

Let $m$ be the least common multiple of the divisors of all the $x_I$, $r=|\min_{I \subseteq [n]} \{ m x_I \} |$, and $k= r 2^{n-1}$. Then, for every $I \subseteq [n]$, by definition of $m$, we have that $mx_I +r$ is an integer and that it is non-negative, since 
\[ \min_{I \subseteq [n]} \{ m x_I + r \} = \min_{I \subseteq [n]} \{ m x_I \} + r  \geq 0 .\]
By linearity and Lemma \ref{lem:novo}, it follows that 
\begin{align*} 
\bm{\blacktriangle_n} (mx+\bm{\overrightarrow{r}}) &=  m \bm{\blacktriangle_n} (x) + \bm{\blacktriangle_n} (\bm{\overrightarrow{r}}) \\
&=  m c + 2^{n-1} \bm{\overrightarrow{r}} \\ 
&=  m c + \bm{\overrightarrow{k}}.
\end{align*}
Thus, by Proposition \ref{pro:realizable}, $(\bm{\blacktriangle_n},(m c + \bm{\overrightarrow{k}}),n)$ is realizable.
\end{proof}

Note that $m c + \bm{\overrightarrow{k}}$ has the same weak ordering as $c$. This is important since, as seen in Proposition \ref{prop:ineq}, two distances are decoding equivalent if they have the same weak ordering. It is in this way that we will use Theorem \ref{teo:sys} to prove Theorems \ref{teo:translation embedding} and \ref{teo:embedding}.

\begin{myexe} \label{exe:teo1}
Lets apply Theorem \ref{teo:sys} to Example \ref{exe:sym}.

We saw that $\bm{\blacktriangle_n} (x) = (3,2,1,3,3,2,3)$ has unique solution
$x = \frac{1}{4} (7,3,1,3,1,1,1)$. From the proof of Theorem \ref{teo:sys} it follows that if we take $m = 2$ and $k = \frac{1}{4}$, then $\bm{\blacktriangle_n} (x) = 2(3,2,1,3,3,2,3)+ \frac{1}{4} \bm{\overrightarrow{1}} $ is realizable. Indeed one can calculate that its solution is $x = (4,2,1,2,1,1,1)$.

Both $\frac{1}{4}(7,3,1,3,1,1,1)$ and $(4,2,1,2,1,1,1)$ have the same weak ordering.
\end{myexe}

\section{Embedding Distances into the Hamming Cube}\label{emb}

Embedding distances isometrically into the Hamming cube is an area of its own \cite{deza97}. Determining if it is possible for a given distance is NP-Hard \cite{chva80}. 

We prove that any semimetric is decoding equivalent to a distance which is isometrically embeddable into the Hamming cube. If, in addition, the semimetric is translation invariant over $\mathbb{F}_2^n$, the embedding is a linear function.

We first note that there is a weight preserving bijection between the $n$-dimensional Hamming cube $H^n$, and the subsets of $[n]$, $2^{[n]}$ given by
\[ supp : H^n \rightarrow 2^{[n]} \]
where $supp(x) = \{ i : x_i \neq 0 \}$.

This function satisfies the following properties:
\begin{enumerate}

\item $supp(x+y) = supp(x) \bigtriangleup supp(y)$
\item $\omega_H (x) = | supp(x) |$.

\end{enumerate}

Thus, isometrically embedding a distance $d$ over $[n]$ into the Hamming cube is equivalent to determining if $(\bm{\blacktriangle_{2}},\delta,n-1)$ is realizable, where 
\[ \delta_{ij} = d(i,j), \hspace{10pt} i,j \in [n-1] \]
\[ \delta_i = d(i,n), \hspace{10pt} i \in [n-1] .\]

By Definition \ref{def:symm} this corresponds to the intersection pattern $(\bm{\Cap_{2}},c,n-1)$:
\[ c_{ij} = \frac{1}{2} ( d(i,n) + d(j,n) - d(i,j) ), \hspace{10pt} i,j \in [n-1] \]
\[ c_i = d(i,n), \hspace{10pt} i,j \in [n-1] .\]

This relation between intersection patterns and Hamming cube embeddings was first pointed out by Deza in \cite{deza73}.

We are interested in embedding up to decoding equivalence. We will first consider the case of translation invariant semimetrics over $\mathbb{F}_2^n$. This will follow directly from Theorem \ref{teo:sys}. Since any semimetric can be seen as translation invariant by adding dummy variables, the general case will follow as a consequence. 

As said earlier, we always assume $I \subseteq [n]$ to be nonempty.

\begin{myteo} \label{teo:translation embedding}
Let $d_1$ be a translation invariant semimetric over $\mathbb{F}_2^n$ with weight $\omega_1$. Then there exists a translation invariant semimetric $d_2$, with weight $\omega_2$ which is decoding equivalent to $d_1$ and is linearly embeddable into the Hamming cube.
\end{myteo}

\begin{proof}
Denote by $\{e_1,e_2,\cdots, e_n\}$ the standard basis of $\mathbb{F}_2^n$.

Let, for every $I \subseteq [n]$,
\[ \delta_I = \omega_1 \left ( \sum_{i \in I} e_i \right )  .\]

Without loss of generality we can assume that $\delta_I \in \mathbb{Q}_+$, since, by Proposition \ref{prop:ineq}, only the order relation between the values matters.

Again by proposition \ref{prop:ineq}, for any given $m,k \in \mathbb{Z}_+$ if , for every $I\subseteq [n]$,
\[ w_2 \left ( \sum_{i \in I} e_i \right ) = m \delta_I + k ,\]
then $d_1 \sim d_2$ since the ordering of the distances is preserved.

Theorem \ref{teo:sys} ensures that there exists $m,k \in \mathbb{Z}_+$ such that $(\bm{\blacktriangle_{n}},m \delta + \bm{\overrightarrow{k}},n)$ is realizable. Thus, there exists a family of sets, $\mathcal{F} = \{ A_1, A_2, \ldots , A_n \}$, such that, for every $I \subseteq [n]$, 
\[ | \triangle_{i \in I} A_i | = m \delta_{I} + k .\]
Let $N=|\cup_i A_i |$ and $f: \mathbb{F}_2^n \rightarrow 2^N$ be such that $f(e_i) = A_i$.
Then, $supp^{-1} \circ f$ is a linear embedding from $(\mathbb{F}_2^n , d_2)$ (where $d_2$ is decoding equivalent to $d_1$) into the $N$ dimensional Hamming cube. The requirement that $d_1$ must be a semimetric is needed for $f$ to be injective.
\end{proof}

We now prove the result for any semimetric.

\begin{myteo}\label{teo:embedding}
Let $d_1$ be a semimetric over $[n]$. Then there exists a distance $d_2$ such that $d_1 \sim d_2$ and $d_2$ is isometrically embeddable into the Hamming cube.
\end{myteo} 

\begin{proof}
Let, for every $I \subseteq [n]$,
\[
\delta_I = \left\{\begin{matrix}
d(i,j) & \text{if} \hspace{4mm} I = \{ i,j \} \hspace{4mm} \\ 
1 & \text{otherwise}
\end{matrix}\right.
\]
By adding these dummy variables, we can now apply Theorem \ref{teo:translation embedding}.
\end{proof}

We now show an example of an embedding of a translation invariant metric into the Hamming cube, using the methods described in Theorem \ref{teo:translation embedding}.

\begin{myexe} \label{exe:embed}

Consider the following translation invariant metric $d$ over $\mathbb{F}_2^3$, with weight $\omega$. 

\[
\begin{matrix}
\omega(001) = 3 \hspace{5mm} & \omega(011) = 3 \hspace{5mm} & \omega(111) = 3 \\ 
\omega(010) = 2 \hspace{5mm} &\omega(101) = 3 \hspace{5mm} & \\ 
\omega(100) = 1 \hspace{5mm} & \omega(110) = 2  \hspace{5mm}& 
\end{matrix}
\]

This corresponds to the following set pattern (which appears in Example \ref{exe:teo1}).

\[
\begin{matrix}
\blacktriangle_1 (x) = 3 \hspace{5mm}& \blacktriangle_{12} (x) = 3 \hspace{5mm}& \blacktriangle_{123}(x) = 3 \\ 
\blacktriangle_{2}(x) = 2 \hspace{5mm}& \blacktriangle_{13}(x) = 3 \hspace{5mm}& \\ 
\blacktriangle_{3}(x) = 1 \hspace{5mm}& \blacktriangle_{23}(x) = 2 \hspace{5mm}& 
\end{matrix}
\]

Using Equation \ref{eq:dif} recursively we get the equivalent set pattern.

\[
\begin{matrix}
\Cap_{1}(x) = 3 \hspace{5mm}& \Cap_{12}(x) = 1 \hspace{5mm}& \Cap_{123}(x) = \frac{1}{4} \hspace{5mm}\\ 
\Cap_{2}(x) = 2 \hspace{5mm}& \Cap_{13}(x) = \frac{1}{2} \hspace{5mm}& \hspace{5mm}\\ 
\Cap_{3}(x) = 1 \hspace{5mm}& \Cap_{23}(x) = \frac{1}{2} \hspace{5mm}& \hspace{5mm}
\end{matrix}
\]

This solution is given by

\[
\begin{matrix}
\vspace{3pt}
x_{1} = \frac{7}{4} \hspace{5mm}& x_{12} = \frac{3}{4} \hspace{5mm}& x_{123} = \frac{1}{4} \hspace{5mm}\\ \vspace{3pt}
x_{2} = \frac{3}{4} \hspace{5mm}& x_{13} = \frac{1}{4} \hspace{5mm}& \hspace{5mm}\\ \vspace{3pt}
x_{3} = \frac{1}{4} \hspace{5mm}& x_{23} = \frac{1}{4} \hspace{5mm}& \hspace{5mm}
\end{matrix}
\]

Taking $x' = 2x + \frac{1}{4} \bm{\overrightarrow{1}} $, by Theorem \ref{teo:sys}, $(\bm{\blacktriangle_{n}},\bm{\blacktriangle_{n}}(x'),n)$ is realizable with
 
\[
\begin{matrix}
\vspace{3pt}
x'_{1} = 4 \hspace{5mm}& x'_{12} = 2 \hspace{5mm}& x'_{123} = 1 \hspace{5mm}\\ \vspace{3pt}
x'_{2} = 2 \hspace{5mm}& x'_{13} = 1 \hspace{5mm}& \hspace{5mm}\\ \vspace{3pt}
x'_{3} = 1 \hspace{5mm}& x'_{23} = 1 \hspace{5mm}& \hspace{5mm}
\end{matrix}
\]

By Theorem \ref{teo:translation embedding}, this corresponds to the linear embedding, $f:\mathbb{F}_2^3 \mapsto H^{12}$
\[ f(100) = 111111110000 \]
\[ f(010) = 000001111110 \]
\[ f(001) = 000011000011 \]
and decoding 
in $(\mathbb{F}_2^3 , d)$ is equivalent to decoding in the image of $f$ in $(H^{12} , d_H)$.

\end{myexe}

\section{Geometric Approach} \label{geo}

The $12$-dimensional embedding in Example \ref{exe:embed} is not optimal, in the sense of minimizing the dimension of the embedding space, i.e. the Hamming cube. We can find an optimal $11$-dimensional embedding by taking a geometric approach.

Consider all the translation invariant semimetrics $d$ over $\mathbb{F}_2^3$, with weight $\omega$. Determining which semimetrics are Hamming cube embeddable is equivalent to determining for which $\delta \in \mathbb{R}^{2^n-1}$ the set pattern $((\bm{\blacktriangle_{n}},\delta,n)$ is realizable. This set pattern  satisfies the following equations

\[
\begin{matrix}
\delta_{A} = x_A + x_{AB} + x_{AC} + x_{ABC} & \hspace{10pt} \delta_{AB} = x_A + x_{B} + x_{AC} + x_{BC}  \\ 
\delta_B = x_B + x_{AB} + x_{BC} + x_{ABC}& \hspace{10pt} \delta_{AC} = x_A + x_{C} + x_{AB} + x_{BC} & \\ 
\delta_C = x_C + x_{AC} + x_{BC} + x_{ABC} & \hspace{10pt} \delta_{BC} = x_B + x_{C} + x_{AB} + x_{AC} \\
\end{matrix}
\]
\[ \delta_{ABC} = x_A + x_{B} + x_{C} + x_{ABC} \]

which can be rewritten as

\[
\delta =
\begin{pmatrix}
\delta_A \\ 
\delta_B \\ 
\delta_C \\ 
\delta_{AB} \\ 
\delta_{AC} \\ 
\delta_{BC} \\ 
\delta_{ABC}
\end{pmatrix}
=
\begin{pmatrix}
1 & 0 & 0 & 1 & 1 & 0 & 1\\ 
0 & 1 & 0 & 1 & 0 & 1 & 1\\ 
0 & 0 & 1 & 0 & 1 & 1 & 1\\ 
1 & 1 & 0 & 0 & 1 & 1 & 0\\ 
1 & 0 & 1 & 1 & 0 & 1 & 0\\ 
0 & 1 & 1 & 1 & 1 & 0 & 0\\ 
1 & 1 & 1 & 0 & 0 & 0 & 1
\end{pmatrix}
.
\begin{pmatrix}
x_{A}\\ 
x_{B}\\ 
x_{C}\\ 
x_{AB}\\ 
x_{AC}\\ 
x_{BC}\\ 
x_{ABC}
\end{pmatrix}.
=Tx
\]

The matrix $T$ can easily be inverted.

\[ T^{-1} =
\begin{pmatrix}
1 & 0 & 0 & 1 & 1 & 0 & 1\\ 
0 & 1 & 0 & 1 & 0 & 1 & 1\\ 
0 & 0 & 1 & 0 & 1 & 1 & 1\\ 
1 & 1 & 0 & 0 & 1 & 1 & 0\\ 
1 & 0 & 1 & 1 & 0 & 1 & 0\\ 
0 & 1 & 1 & 1 & 1 & 0 & 0\\ 
1 & 1 & 1 & 0 & 0 & 0 & 1
\end{pmatrix}^{-1}
=
\frac{1}{2^{2}}
\begin{pmatrix}
1 & -1 & -1 & 1 & 1 & -1 & 1\\ 
-1 & 1 & -1 & 1 & -1 & 1 & 1\\ 
-1 & -1 & 1 & -1 & 1 & 1 & 1\\ 
1 & 1 & -1 & -1 & 1 & 1 & -1\\ 
1 & -1 & 1 & 1 & -1 & 1 & -1\\ 
-1 & 1 & 1 & 1 & 1 & -1 & -1\\ 
1 & 1 & 1 & -1 & -1 & -1 & 1
\end{pmatrix}
\]

If we look at $\mathbb{R}_{ > 0}^7$ as the space of all translation invariant semimetrics, $T$ is an automorphism of this space, and the set of semimetrics which can be embedded into the Hamming cube is the integer cone
$\{ T x : x \in \mathbb{Z}_{\geq 0} \}$\footnote{It is known that Hamming cube embedable distances form an integer cone \cite{deza97}.}.

Since we are interested in embeding up to decoding equivalence we must see how this equivalence partitions the space. By Proposition \ref{prop:ineq}, each equivalence class will correspond to a weak ordering on the $7$ coordinates. This partitions the space into $13$ cones ($6$ $3$-dimensional, $6$ $2$-dimensional, and $1$ $1$-dimensional).

The dimension of an embedding is the $L_1$ weight of the vector corresponding to the minterms, i.e. if $\delta \in \{ T x : x \in \mathbb{Z}_{\geq 0} \}$,
then the dimension of the embedding is $|x|_1 = x_A + x_B + x_C + x_{AB} + x_{AC} + x_{BC} + x_{ABC}$.

If we denote by $Cone (\delta) = \{ \delta' \in \mathbb{R}_{ \geq 0}^7 : \delta \sim \delta' \}$, 
then finding the minimal dimensional embedding of a distance $\delta \in \mathbb{R}_{ \geq 0}^7$ (up to decoding equivalence) is an integer linear programming problem of the form
\[
\begin{matrix}
\text{Minimize:} & |x|_1\\ 
\text{subject to:} & \hspace{20pt} Tx \in Cone(\delta)
\end{matrix}
\]

\begin{myexe}
Consider the same translation invariant metric $d$ over $\mathbb{F}_2^3$, with weight $\omega$, as in Example \ref{exe:embed}. Let us find the minimal dimensional embedding (up to decoding equivalence) into the Hamming cube.

As we saw previously we can write the system as
\[
\delta =
\begin{pmatrix}
3\\ 
2\\ 
1\\ 
3\\ 
3\\ 
2\\ 
3
\end{pmatrix}
=
\begin{pmatrix}
1 & 0 & 0 & 1 & 1 & 0 & 1\\ 
0 & 1 & 0 & 1 & 0 & 1 & 1\\ 
0 & 0 & 1 & 0 & 1 & 1 & 1\\ 
1 & 1 & 0 & 0 & 1 & 1 & 0\\ 
1 & 0 & 1 & 1 & 0 & 1 & 0\\ 
0 & 1 & 1 & 1 & 1 & 0 & 0\\ 
1 & 1 & 1 & 0 & 0 & 0 & 1
\end{pmatrix}
.
\begin{pmatrix}
x'_{A}\\ 
x'_{B}\\ 
x'_{C}\\ 
x'_{AB}\\ 
x'_{AC}\\ 
x'_{BC}\\ 
x'_{ABC}
\end{pmatrix}.
\]
Solving it we get
\[x' =
\begin{pmatrix}
x'_{A}\\ 
x'_{B}\\ 
x'_{C}\\ 
x'_{AB}\\ 
x'_{AC}\\ 
x'_{BC}\\ 
x'_{ABC}
\end{pmatrix}
=
\frac{1}{4}
\begin{pmatrix}
7\\ 
3\\ 
1\\ 
3\\ 
1\\ 
1\\ 
1
\end{pmatrix}
.\]
We are searching for $x \in \mathbb{Z}_+$ which minimizes $|x|_1$, and such that $Tx \in Cone(\delta)$. By Proposition \ref{prop:ineq}, $x$ must satisfy the same weak ordering on its coordinates as $x'$, i.e.
\[ x_{ABC}=x_{BC}=x_{AC}=x_{C}<x_{AB}=x_{B}<x_{A} .\]
Thus, to minimize $|x|_1$ we must take
\[ x_{ABC}=x_{BC}=x_{AC}=x_{C} = 1 ,\]
\[ x_{AB}=x_{B} = 2 \hspace{5pt} \text{and} \hspace{5pt} x_{A} = 3.\]

This corresponds to the following Venn Diagram.

\begin{center}
\includegraphics[width=6cm]{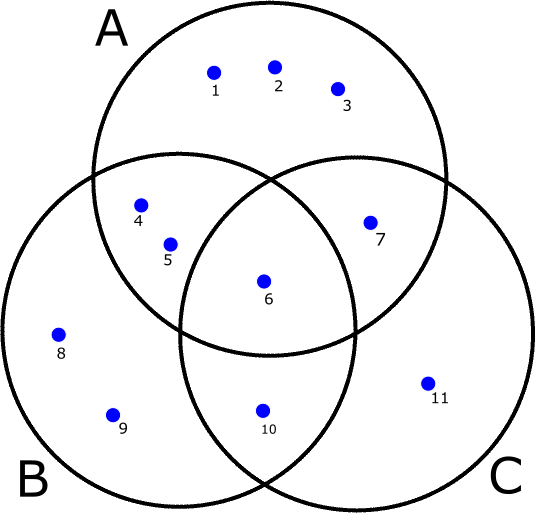}
\end{center}

The corresponding linear embedding $f: \mathbb{F}_2^3 \mapsto H^{11}$ is such that
\[ 001 \mapsto 11111110000 \]
\[ 010 \mapsto 00011101110 \]
\[ 100 \mapsto 00000110011 .\]
This embedding has lower dimension than the one in Example \ref{exe:embed} and it is optimal.
\end{myexe}

The case just discussed is typical. The $n$-dimensional case will behave equally. 

In the case of general distances, however, this geometric analysis requires more care. The equivalence class will still partition the space into cones, but they will not be as simple as for the translation invariant case.

\section{Future Work}

We list some possible topics for future work:
\begin{itemize}
\item Consider important families of channels (like the binary assymetric channel) and see under which conditions they are metriziable.
\item Investigate which properties are invariant under the decoding equivalence (like a code being perfect).
\item Use Hamming embeddings to better understand these properties (perfect codes are much better understood over the Hamming metric than for other metrics in general).
\item Count the number of distances up to decoding equivalence.
\end{itemize}

We also believe that the geometric approach might be useful in investigating the problems above.

\renewcommand{\abstractname}{Acknowledgements}
\begin{abstract}
Rafael G.L. D'Oliveira was supported by CAPES.

Marcelo Firer was partially supported by grant 2013/25977-7, S\~ao Paulo Research Foundation, (FAPESP) and grant 303985/2014-3, CNPq.

Both authors would like to thank Michel Deza.
\end{abstract}

\end{document}